\documentclass{sig-alternate_tight_title_no_copyright}

\usepackage{amssymb}
\usepackage{graphicx}
\usepackage{url}

\usepackage{amsmath}
\usepackage{amssymb}
\usepackage{amsfonts}
\usepackage{float}
\usepackage{graphics}
\usepackage{theorem}
\usepackage{euscript}
\usepackage{psfrag}

\newcommand{\abs}[1]{\left\vert#1\right\vert}

\newcommand{\Mat}[1]{\boldsymbol{#1}}

\newtheorem{thm}{Theorem}

\newtheorem{prop}[thm]{Proposition}

\def\E{{\mathbb E}}        

\def\myscale{0.65}

\newcommand{\comment}[1]{}
\newcommand{\paren}[1]{\left( #1 \right)}
\newcommand{\bracket}[1]{\left[ #1 \right]}

\begin{document}




\title{On the Impact of Random Index-Partitioning on Index Compression}


\numberofauthors{3}
\author{
\alignauthor M. Feldman\\
       \affaddr{CS Department, Technion}\\
       \affaddr{Haifa 32000, Israel}\\
        \email{moranfe@cs.technion.ac.il}
\alignauthor
R. Lempel\\
       \affaddr{Yahoo! Labs}\\
       \affaddr{Haifa 31905, Israel}\\
       \email{rlempel@yahoo-inc.com}
\alignauthor O. Somekh\\
       \affaddr{Yahoo! Labs}\\
       \affaddr{Haifa 31905, Israel}\\
       \email{orens@yahoo-inc.com}
\and 
\alignauthor K. Vornovitsky\\
       \affaddr{CS Department, Technion}\\
       \affaddr{Haifa 32000, Israel}\\
       \email{kolman@cs.technion.ac.il}
}

\maketitle

\begin{abstract}

The performance of processing search queries depends heavily on the stored index size. Accordingly, considerable research efforts have been devoted to the development of efficient compression techniques for inverted indexes. 
Roughly, index compression relies on two factors: the ordering of the indexed documents, which strives to position similar documents in proximity, and the encoding of the inverted lists that result from the ordered stream of documents.
Large commercial search engines index tens of billions of pages of the ever growing Web. The sheer size of their indexes dictates the distribution of documents among thousands of servers in a scheme called local index-partitioning, such that each server indexes only several millions pages. Due to engineering and runtime performance considerations, random distribution of documents to servers is common. However, random index-partitioning among many servers adversely impacts the resulting index sizes, as it decreases the effectiveness of document ordering schemes.\\
We study the impact of random index-partitioning on document ordering schemes. We show that index-partitioning decreases the aggregated size of the inverted lists logarithmically with the number of servers, when documents within each server are randomly reordered. On the other hand, the aggregated partitioned index size increases logarithmically with the number of servers, when state-of-the-art document ordering schemes, such as lexical URL sorting and clustering with TSP, are applied.
\comment{Focusing on URL sorting, we observe that preserving the order across web hosts is not crucial, and similar index sizes are achieved with intra-host URL sorting only. Moreover, we show that clustering pages by host names without regard to intra-host ordering already brings significant compression benefits.}
Finally, we justify the common practice of randomly distributing documents to servers, as we qualitatively show that despite its ill-effects on the ensuing compression, it decreases key factors in distributed query evaluation time by an order of magnitude as compared with partitioning techniques that compress better.

\vspace{-0.1cm}
\vspace{-0.1cm}
\vspace{-0.1cm}
\end{abstract}

\\
\section{Introduction}
\label{sec: introduction}

The searchable Web spans tens of billions of pages, yet search engine users expect fresh and relevant search results to be delivered within less than a second. Serving simultaneously thousands of queries, web search engines use an {\em inverted index}, a data structure that supports efficient retrieval of documents containing a set of terms given by the user's query. Due to the huge number of web pages and the resulting  amount of data, the index is partitioned over thousands of servers, where each server typically stores and processes the inverted index of only several millions documents \cite{Arasu-etal_ACM-Internet-Tech_2001,Barroso2003,Zobel-Moffat-ACM_CSUR2006,IIR}. 
At query time, the query is sent to all servers for processing, and the top results retrieved from all servers are merged to produce the final results, which are returned to the user.

The inverted index data structure contains a {\em postings list} for each unique term appearing in the corpus. The postings list of term $t$ consists of the list of document identifiers\footnote{Although terms frequencies and offsets within the document occupy a major portion of modern inverted indexes, we focus here on the documents identifiers only.} (docIds) containing $t$. The documents within each list are typically sorted by increasing docIds values, and the list is represented by encoding the gaps (called {\em dGaps}) between successive docIds. Another data structure in an inverted index is the {\em lexicon}, or {\em dictionary}, which is a lookup table that for each term $t$ in the corpus, points to the postings list corresponding to $t$ \cite{Arasu-etal_ACM-Internet-Tech_2001,Zobel-Moffat-ACM_CSUR2006,ModernIR2}.

\comment{
The benefits bestowed by inverted index size reduction have encouraged a vast research effort along the two system degrees of freedom mentioned above, i.e., docId assignment \cite{Blandford-Blelloch_DCC-2002}\ -\nocite{Shieh-Chen-Shann-Chung-IPM2003}\nocite{Blanco-Barreiro-IR2006}
\nocite{Blanco-Barreiro-IR2006}
\nocite{Yan-Ding-Suel-WWW2009}\cite{Ding-Attenberg-Suel-WWW2010} and dGap compression \cite{Anh-Moffat-IR2005}\ -\ \nocite{Boldi-Vigna-InternetMath2005}\cite{Zhang-Long-Suel-WWW2008}.
}
      
Index size has an important effect on system performance. In addition to the direct reduction in memory and disk space, more compact indexes lead to savings in I/O transfers and increase the hit rate of memory caches, offering an improvement in overall query processing throughput \cite{Zhang-Long-Suel-WWW2008,Yan-Ding-Suel-WWW2009}. Consequently, a large body of work has focused on index compaction and compression methods. The structure described above leaves two main degrees of freedom for compression optimization: (a) the assignment of docIds to documents (also referred to as document reordering); and (b) the actual encoding of the dGaps into bits (also referred to as dGap compression  \cite{witten1,Moffat-Stuiver-IR2000,Boldi-Vigna-InternetMath2005,Haman-Thesis2005,
Anh-Moffat-IR2005,Zhang-Long-Suel-WWW2008}). This work focuses on the former.

The basic idea behind an effective docId assignment is to place ``similar'' documents close to each other, hence, potentially reducing the dGaps since similar documents contain many common terms. Such effective assignment produces highly clustered posting lists where long ``runs'' of small dGaps are separated by large dGaps. In contrast, a random assignment of docIds would result in dGaps that approximately follow a Geometric distribution within each postings list \cite{Chierichetti-Kumar-Raghavan-WWW2009}. The problem of finding the optimal docId assignment can be explicitly expressed in closed form, but is, unfortunately, NP-hard \cite{Gonzalez-PhD-2008}. Therefore, most works on document assignment proposed various heuristics that includes approximations to the traveling salesman problem (TSP), solutions based on clustering algorithms, and solutions based on the natural URL lexicographical ordering of web pages \cite{Blandford-Blelloch_DCC-2002,Shieh-Chen-Shann-Chung-IPM2003,Silvestri-Perego-Orlando-SAC2004,
Blanco-Barreiro-IR2006,Silvestri-ECIR2007,Yan-Ding-Suel-WWW2009,Ding-Attenberg-Suel-WWW2010}.

All aforementioned efforts focused on compacting the inverted index of a single server. However, large  corpora are indexed over thousands of servers, with each server handling only several million documents 
\cite{Arasu-etal_ACM-Internet-Tech_2001,Barroso2003,Zobel-Moffat-ACM_CSUR2006,IIR}. 
In order to better balance the number of result documents resulting on each server, thereby decreasing query processing time (see our experiments in Section~\ref{sec:queryproc}),
documents are often distributed randomly among the servers
\cite{Arasu-etal_ACM-Internet-Tech_2001,Barroso2003,prefetching,IIR}. 
As first noted in \cite{Yan-Ding-Suel-WWW2009}, this index-partitioning operation may have a profound effect on document assignment algorithms, since similar documents (e.g., pages of the same web host) are often routed to different servers. This work examines the impact of random index-partitioning on the effectiveness of docId assignment algorithms that aim to compress the inverted index. 
Our experiments are performed on the 25 million web page TREC .gov2 collection. Our main contributions are the following:
\begin{itemize}
    \item We showcase the interplay between \textit{random} index-parti-tioning and compression. \comment{As far as we are aware, this novel problem has not been studied previously\footnote{We note that it was raised as a concern in \cite{Yan-Ding-Suel-WWW2009}. \textbf{<} Another relevant work in this context is \cite{Lavee_Liberty_Lempel_Somekh-WWW2011-UNPUB}, where \textit{non-random} online document routing schemes are used for index compactness purposes \textbf{>}.}.}
    \item We quantitatively and analytically show that the performance gap between effective docId assignment heuristics and ineffective ones diminishes as the index is \textit{randomly} partitioned over more servers. For example, with dGap Delta encoding, the total length of the inverted lists actually decreases logarithmically with the number of partitions when docIds are assigned randomly. On the other hand, partitioning causes that size to increase logarithmically with the number of partitions when effective docId assignments such as URL sorting, and clustering with TSP, are applied. Similar trends are reported for dGap block PForDelta encoding as well.
    \item We study experimentally the factors that make the URL-based assignment perform well in practice. We show that inter-host ordering hardly matters, and that clustering pages by hosts with arbitrary intra-host ordering already brings significant compression benefits. \comment{Furthermore, we define a distance function of a docId assignment from the URL-based assignment, and show that this distance predicts well, in some sense, the performance of the docId assignment.}
\item 
We justify the common practice of randomly distributing documents to servers, as we qualitatively show that despite its ill-effects on the ensuing compression, it decreases key factors in distributed query evaluation time by an order of magnitude as compared with the better compressing URL-based partitioning.
\comment{
We study the impact of random and URL-based partitioning of documents to servers and demonstrate the benefits of the former in terms of query processing time. We quantitatively and analytically explain the behavior of the query processing time for the various partitioning schemes as function of the number of partitions \textbf{SIGIR}.  
}

The rest of this work is organized as follows. Section~\ref{sec: background} provides background and surveys related work. The experimental setup is described in section~\ref{sec: setup}. Experimental results and analytical insight of the impact of partitioning on index sizes are reported in Sections~\ref{sec: results} and 
\ref{sec: Analytical Insight on Results}, respectively. The impact of index partitioning on query processing time is considered in Section \ref{sec:queryproc}. Finally, we conclude in Section \ref{sec: conc}.  
\end{itemize}

\section{Background and Prior Work}\label{sec: background}

\subsection{Index Partitioning}\label{sec: Index-Partitioning}

\comment{
\begin{figure}
\centering
\includegraphics[scale=0.4]{Figure_Partitioning-Strategies.eps}
\caption{Document- and term-based partitioning of a five documents sample corpus over a cluster of three nodes.}
\label{fig: Index-Partitioning}
\end{figure}
}
The sheer size of the Web, the enormous number of search queries, and the required low latency, enforce a distributed inverted index architecture \cite{Arasu-etal_ACM-Internet-Tech_2001,Zobel-Moffat-ACM_CSUR2006,Barroso2003}. 
To support these requirements, both \textit{distribution} and \textit{replication} principles are applied. Replication (or \textit{mirroring}) means making enough identical copies of the system so that the required query load can be served, and is beyond the scope of this work. Distribution means the way the inverted index is partitioned across a collection of nodes.

The two main strategies of partitioning an inverted index are \textit{local index-partitioning} and \textit{global index-partitioning} \cite{Badue-BaezaYates-RibeiroNeto-Ziviani-SPIRE2001,Moffat-Webber-Zobel-BaezaYates-INFIR2007}. 
According to the local index-partitioning strategy (or \textit{document based partition}), each node is responsible for a disjoint subset of documents in the collection. Each search query is sent to all nodes, each of which returns its top ranking documents for the query. Those lists are then combined in some way to provide the end result. In the global index-partitioning strategy (or \textit{term based partition}), terms are divided into disjoint subsets, such that each node stores postings lists only for a subset of terms.

Due to various theoretical and practical considerations, large-scale search engines follow the local inverted index-partitioning strategy distributing documents across the nodes \cite{Barroso2003,Badue-BaezaYates-RibeiroNeto-Ziviani-SPIRE2001,Moffat-Webber-Zobel-BaezaYates-INFIR2007}. Documents can be distributed to nodes using different policies. For example, the \textit{hash distribution policy} allocates documents to nodes in a random fashion by hashing the documents' URLs to yield a node identifier \cite{prefetching, IIR}. Other  policies such as \textit{round-robin distribution} are also possible \cite{Hirai-Raghavan-GarciaMolina-Paepcke-WWW2000}.

While random distribution of documents to nodes is used by commercial search engines \cite{Arasu-etal_ACM-Internet-Tech_2001,Barroso2003,prefetching, IIR}, other distribution schemes were considered in distributed information retrieval systems and peer-to-peer networks. For instance, in \cite{Xu-Croft_SIGIR1999} (see also \cite{Kulkarni_Callan-LSDSIR-2010} for a more recent work) the authors used a two-pass K-means clustering algorithm and a KL-divergence distance metric to organize a document collection into 100 topical clusters (or \textit{shards}) and demonstrated the benefits of selectively searching only a few shards per query. Query logs were used by the authors of \cite{Puppin-Silvestri-Laforenza_InfoScale2006} (see also \cite{Puppin-Silvestri-Perego-BaezaYates_ACM-TIS2010} for a more recent work) to organize a document collection into multiple shards. Selectively searching shards defined by these clusters was found to be more effective than selectively searching randomly defined shards. Non-random distribution of documents in a distributed search engine was recently considered in \cite{Lavee_Liberty_Lempel_Somekh-WWW2011-UNPUB}, where the authors treat the routing of documents to nodes as an online problem in an incremental indexing setting. Under a model where routed documents are appended to the existing index partitions, the authors demonstrate a tradeoff between the compression of a locally-partitioned index and the balanced distribution of documents from the same host across the index partitions.

\subsection{Inverted Index Compression}

As mentioned in the Section \ref{sec: introduction}, we consider a simplified model of an inverted index in which the postings list of term $t$ holds the 
docIds containing $t$, sorted by increasing value. Denote the list by $d^t_1, d^t_2, \ldots, d^t_{n_t}$, where $d^t_i$ denotes the docId of the $i$'th document containing $t$ out of $n_t$ such documents. The list is actually represented by encoding the first docId and the sequence of gaps (dGaps) between successive identifiers thereafter, i.e. $d^t_1, d^t_2-d^t_1, \ldots, d^t_{n_t}-d^t_{n_t-1}$.
\comment{
Here we focus on a single node of a local index-partitioning architecture, and describe its structure which is dictated by the requirements mentioned earlier.

Let $D=\{d_1,\ d_2,\ldots,\ d_{\abs{D}}\}$ be a sub-collection of $\abs{D}$ documents, and let $T=\{t_1,\ t_2,\ldots,\ t_{\abs{T}}: t_i\in D\}$ be the set of all terms which appear in that sub-collection $D$. A local inverted-index structure is an optimized instantiation of the term-document (termDoc) matrix $\Mat{H}$ where each term and document correspond to a row and column respectively. Hence, $\Mat{H}_{i,j}$ represents the association between term $i$and document $j$, e.g., the term frequency or a binary indicator. Since most document contain only a small subset $T$, and most terms appear in a small subset of $D$, the termDoc matrix $\Mat{H}$ is normally extremely sparse. The inverted-index stores the non-zero entries of $\Mat{H}$ using \textit{posting lists}, where the list corresponds to term $i$ includes the columns numbers of the non-zero entries of found in the $i$th row of $\Mat{H}$. More precisely, the posting lists stores the first number and the gaps (or dGaps) between successive numbers hereafter.
}
The two degrees of freedom available for compressing the size of the lists are (a) docId assignment; and (b) dGap encoding. As we focus on the former, we start by briefly reviewing the latter. dGap encoding techniques aim to compress a sequence of integers. 
\comment{
We briefly mention several techniques in the sequel.

\paragraph{Gamma Encoding:} for $x\in Z^+$ represent $x-2^{\lfloor \log_2x\rfloor}$ in binary and prefix this by unary representation of the binary length of $x$. The number of bits used is given by $S_\gamma(x)=1+2\lfloor \log x\rfloor$. This encoding performs well for very small numbers, but is not appropriate for larger numbers.

\paragraph{Delta Encoding:} for $x\in Z^+$ represent $x-2^{\lfloor \log_2x\rfloor}$ in binary and prefix this by Gamma encoding of the binary length of $x$. The number of bits used is given by 
\begin{equation}\label{eq: delta dgap encoding length}
    S_\delta(x) = 1+\lfloor\log_2x \rfloor + \log_2(1+\lfloor\log_2 x\rfloor)\ ,
\end{equation}
This encoding performs better than Gamma encoding for large numbers.

\paragraph{PForDelta \cite{Haman-Thesis2005}:} first determines a value $b$ such that most (say 90\%) dGaps are less than $2^b$ and thus can be expressed by $b$ bits each. The remaining values are then encoded separately. This technique achieves fast decoding while maintaining small compressed size.  
}
The literature contains schemes that encode each gap individually, e.g. Gamma, Delta, Golomb-Rice \cite{witten1} and Zeta \cite{Boldi-Vigna-InternetMath2005} encodings, as well as schemes that encode certain blocks of gaps, e.g. PForDelta \cite{Zhang-Long-Suel-WWW2008,Haman-Thesis2005} and Simple9 \cite{Anh-Moffat-IR2005}. Additionally, the Interpolative Encoding scheme \cite{Moffat-Stuiver-IR2000} is applied directly on the docIds rather than their dGaps, and works well for clustered term occurrences. 


In general, the docId assignment problem seeks a permutation of the 
documents that minimizes the inverted-index size under a specific dGap encoding scheme. As shown in \cite{Gonzalez-PhD-2008}, this problem is NP-hard and various heuristics are used to provide approximations.    

The size of an inverted-index is a function of the dGaps, which themselves depend on the way docIds are assigned to documents. All effective dGap encoding techniques represent smaller numbers with fewer bits (about logarithmic in the number value). Hence, assigning docIds in a way which results in smaller dGaps is the key for better compression. This principle drives most works dealing with docId assignment, which accordingly strive to assign close docIds to ``similar'' documents, i.e. documents that share many terms.

Technically, most works define a graph $G=(D,E)$, where $D$ is the set of documents, and $E$ is a set of edges representing the similarity between two documents $d_i,d_j\in D$. One line of work started by \cite{Shieh-Chen-Shann-Chung-IPM2003} traverses the graph $G$ to find the maximal weight path connecting all the nodes, assigning docIds accordingly. This is equivalent to the NP-Hard \textit{traveling salesman problem} (TSP). Several TSP approximations were applied for docId assignment in \cite{Shieh-Chen-Shann-Chung-IPM2003,Blanco-Barreiro-IR2006,Ding-Attenberg-Suel-WWW2010}. In \cite{Shieh-Chen-Shann-Chung-IPM2003}, a simple greedy nearest neighbors (GNN) approach is used to add one edge at a time. To reduce the computational load, \cite{Blanco-Barreiro-IR2006} uses singular value decomposition (SVD) to reduce the dimensionality of the term-document matrix. To scale up TSP-based schemes \cite{Ding-Attenberg-Suel-WWW2010} proposes a new framework based on computing TSP on a reduced sparse graph obtained through \textit{locality sensitive hashing}.

In yet another line of work, the nodes of $G$ are clustered according to their similarity and close docIds are assigned to the nodes (documents) within each cluster. A top-down approach is used in \cite{Blandford-Blelloch_DCC-2002}, where the whole collection is recursively split into sub-collections, inserting ``similar'' nodes into the same sub-collections. Then, the sub-collections are merged into an ordered group of nodes. A bottom-up approach called k-scan was proposed in \cite{Silvestri-Perego-Orlando-SAC2004}. A hybrid method which combines k-scan clustering and TSP for intra-cluster docId assignment is proposed by \cite{Gonzalez-PhD-2008}, and will be used in the experiments reported in this paper.

A different approach, which is both highly scalable and highly effective, was proposed for Web collections in \cite{Silvestri-ECIR2007}. It assigns docIds according to the lexicographically sorted order of the documents' URLs, utilizing the fact that URL similarity is a strong indicator of document similarity. The scheme was found to perform remarkably well on various Web collections indexed as a whole. It was not, to the best of our knowledge, examined for partitioned collections.


In all the aforementioned works, a heuristic of docId assignment or an encoding of dGaps were empirically tested against several collections and compared to the results of other works. In contrast, \cite{Chierichetti-Kumar-Raghavan-WWW2009} analyzes the compressibility of a collection whose documents are generated by a simple probabilistic model in which terms are chosen independently from a given distribution.

\section{Experimental Setup}\label{sec: setup}

We use the TREC .gov2 Web corpus, a collection of about 25.2 million pages crawled from the \textsf{gov} domain, for the experiments.
\comment{
We note that since it was crawled almost to exhaustion, GOV2 is very ``dense'' \cite{Ding-Attenberg-Suel-WWW2010}. For some of the experiments we use only part of .gov2. In such cases we take the required number of documents according to their URL lexicographical order (i.e. URL ``continuous'').
}
After parsing, tokenizing (with standard English \textit{stopward removal} and no \textit{stemming}) and removing all empty documents, we are left with 24.9 million documents, 74.5 million distinct terms, and 5,705.2 million postings (distinct term appearances in documents). Whenever we partition the corpus over $m$ servers, documents are assigned to servers independently and uniformly at random. We then apply some docId assignment and dGap encoding schemes across all servers. 
\comment{We note that while Lucene is used for its text processing capabilities and other utilities, we do not report Lucene index sizes in the sequel. Rather, index sizes are reported using the {\em bits per posting} metric, defined below.}
Index sizes are reported using the {\em bits per posting} metric, defined below.

\subsection{The Bits per Posting Metric}

Let a corpus with ${\cal N}$ overall postings be indexed across $m$ partitions, and let $T_i$ denote the set of distinct terms on the $i$'th partition. Let $t$ be a term appearing in $n_t$ documents in some partition, and denote those docIds by $1\le d^t_1 < d_2^t <\ldots <d^t_{n_t}$. Then, the overall size of all postings lists on the $i$'th partition, ${\cal P}_i$, is given by
\[ {\cal P}_i =  \sum_{t \in T_i} S\left(d_1^t,d_2^t-d_1^t,\ldots,d^t_{n_t}-d^t_{n_{t}-1}\right) ,\]
where $S(\cdot)$ is the length (in bits) of encoding the given integer sequence. The overall size of the postings across the $m$ partitions, ${\cal P}$, is defined as 
\[ {\cal P} = \sum_{i=1}^m {\cal P}_i\ . \]
We experiment with Delta and PForDelta encoding schemes. For Delta encoding
\[ S\left(d_1^t,d_2^t-d_1^t,\ldots,d^t_{n_t}-d^t_{n_{t}-1}\right)=\delta (d_1^t) + \sum_{j=2}^{n_t} \delta(d_j^t-d_{j-1}^t)\ ,\]
where $\delta(k)$ is the length (in bits) of the Delta encoding of the positive integer $k$:
\[ \delta(k) = 1+\lfloor\log_2k \rfloor + 2\lfloor\log_2(1+\lfloor\log_2 k\rfloor)\rfloor\ .\]
For PForDelta encoding scheme, each posting list is processed according to the scheme presented in \cite{Zhang-Long-Suel-WWW2008}, with block length of 128 dGaps, and threshold of 90\%. Shorter blocks at the end of long posting lists and short posting lists, down to 64 dGaps are encoded in a similar fashion, while blocks of less than 64 dGaps are simply Delta encoded.

We further define the overhead ${\cal OH}$ of a partitioned index as the space taken by the $m$ dictionaries of the individual partitions. Each entry of the $i$'th dictionary is a pointer into the sequence of postings lists on the $i$'th server, and hence requires $\log_2 {\cal P}_i$ bits\footnote{For simplicity, we assume that individual posting, as well as inverted lists, can start on arbitrary bit boundaries.}. Overall,
\[ {\cal OH} = \sum_{i=1}^m \vert T_i \vert \log_2{{\cal P}_i}\ .\]
Finally, the {\em bits per posting} metric comes in two flavors, with and without overhead. Those are simply $\frac{{\cal P}+{\cal OH}}{{\cal N}}$ and $\frac{{\cal P}}{{\cal N}}$, respectively.

\subsection{DocId Assignment Schemes}
As stated earlier, we are mainly interested in two aspects: (a) studying the impact of random index-partitioning on the bits per posting metric, and (b) gaining further insight into the power of the URL-based docId assignment scheme. We thus experiment with the following five docId assignment schemes:
\begin{description}
\item[Random assignment (RND):] this method serves as a baseline for comparison purposes.
\item[URL-based sorting (URL):] following \cite{Silvestri-ECIR2007}, the documents are sorted lexicographically based on their URL\footnote{The host name components are first inverted, see \cite{Silvestri-ECIR2007} for details.} and docIds are assigned accordingly. 
\item[Clustering assignment (KSCN-TSP):] we adopt a procedure presented in \cite{Blanco-Barreiro-IR2006} where each server's collection is partitioned into $K$ clusters, and GNN approximation of TSP is used to assign the docIds within each cluster. We set the cluster size (and the number of clusters) to the square root of the server's corpus size, which is known to provide fair results. This heuristic represents, in this work, the state-of-the-art of schemes that are URL-agnostic.
\item[Intra-host URL-based sorting (IH-URL):] here, the hosts are randomly ordered, and URL-based ordering is kept within the hosts only. This scheme, when compared to the conventional URL scheme, should reveal the contribution of the inter-host lexicographical ordering to the power of URL-based docId assignment.
\item[Intra-host random assignment (IH-RND):] here, documents of the same host are assigned with consecutive docIds, but both the hosts and the documents within each host are randomly ordered. This scheme should reveal whether the power of URL-based assignment stems merely from the fact that documents of the same host are clustered together, or actually depends on the lexicographic ordering within each host.
\end{description} 

\comment{
Our experiments are composed of the following main phases:
\begin{itemize}
\item \textbf{Corpus generation}: a preliminary phase of indexing .gov2 by Lucene. In case a smaller corpus is required, the resulting Lucene full index is used to produce the new index where URL continuous bulk of documents are selected. 

\item \textbf{Index-partitioning:} We use a random and round-robin local index partitioning, and generate the corresponding Lucene sub-indexes.

\item \textbf{docIds assignment:} for each of the $m$ sub-collections we generate the document reordering vectors according to the five docId assignment schemes mentioned earlier (i.e., RND, URL, KSCN-TSP, ID-URL, and ID-RND). 
}
When comparing URL-agnostic docId assignment schemes (represented here by KSCN-TSP) to the URL sorting scheme over partitioned indexes, one hypotheses comes to mind:
URL-agnostic schemes should outperform URL assignment when the corpus is highly partitioned, since they have the degree of freedom to arrange documents by similarity that transcends diluted URL patterns. 

\section{Experimental Results}\label{sec: results}

\begin{figure*}[t]
\centering
\includegraphics[scale=\myscale]{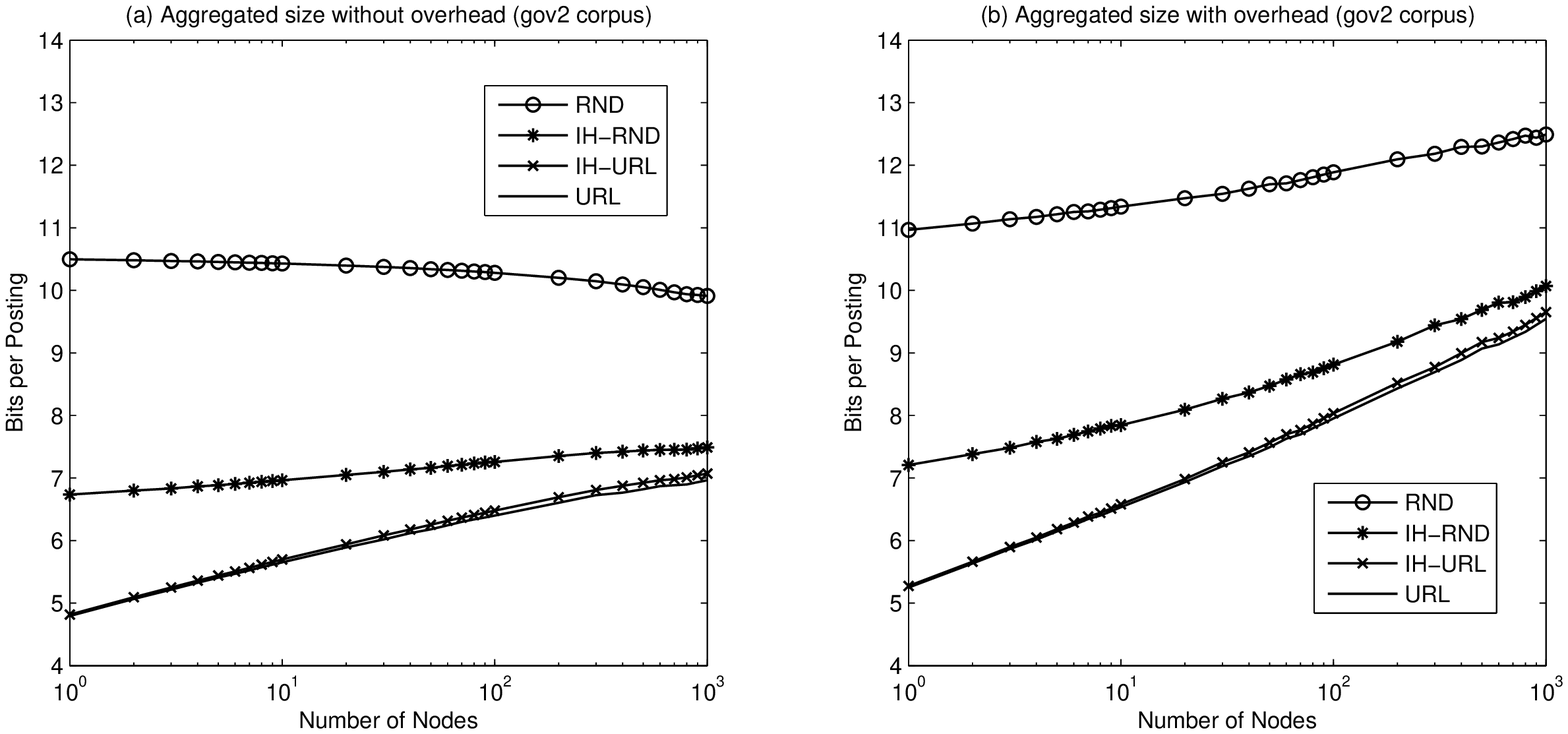}
\caption{Bits per posting as function of the number of nodes for different docId assignment schemes and Delta encoding applied to .gov2 corpus, (a) without and (b) with overhead.}
\label{fig: gov2}
\end{figure*}

\begin{figure*}[t]
\centering
\includegraphics[scale=\myscale]{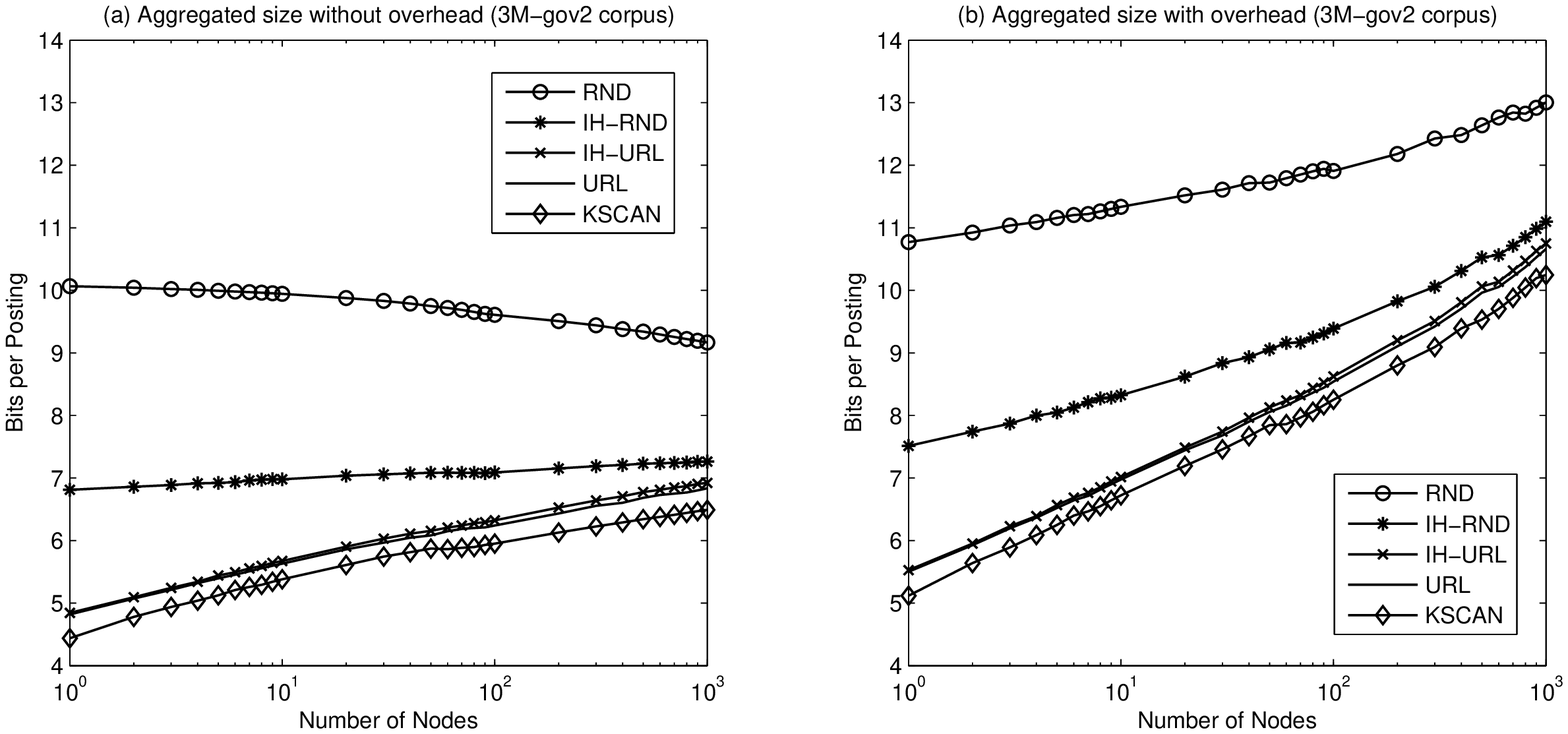}
\caption{Bits per posting as function of the number of nodes for different docId assignment schemes and Delta encoding applied to a bulk of 3 million URL-continuous documents from .gov2 corpus, (a) without and (b) with overhead.}
\label{fig: gov2-3M}
\end{figure*}

\begin{figure*}[t]
\centering
\includegraphics[scale=\myscale]{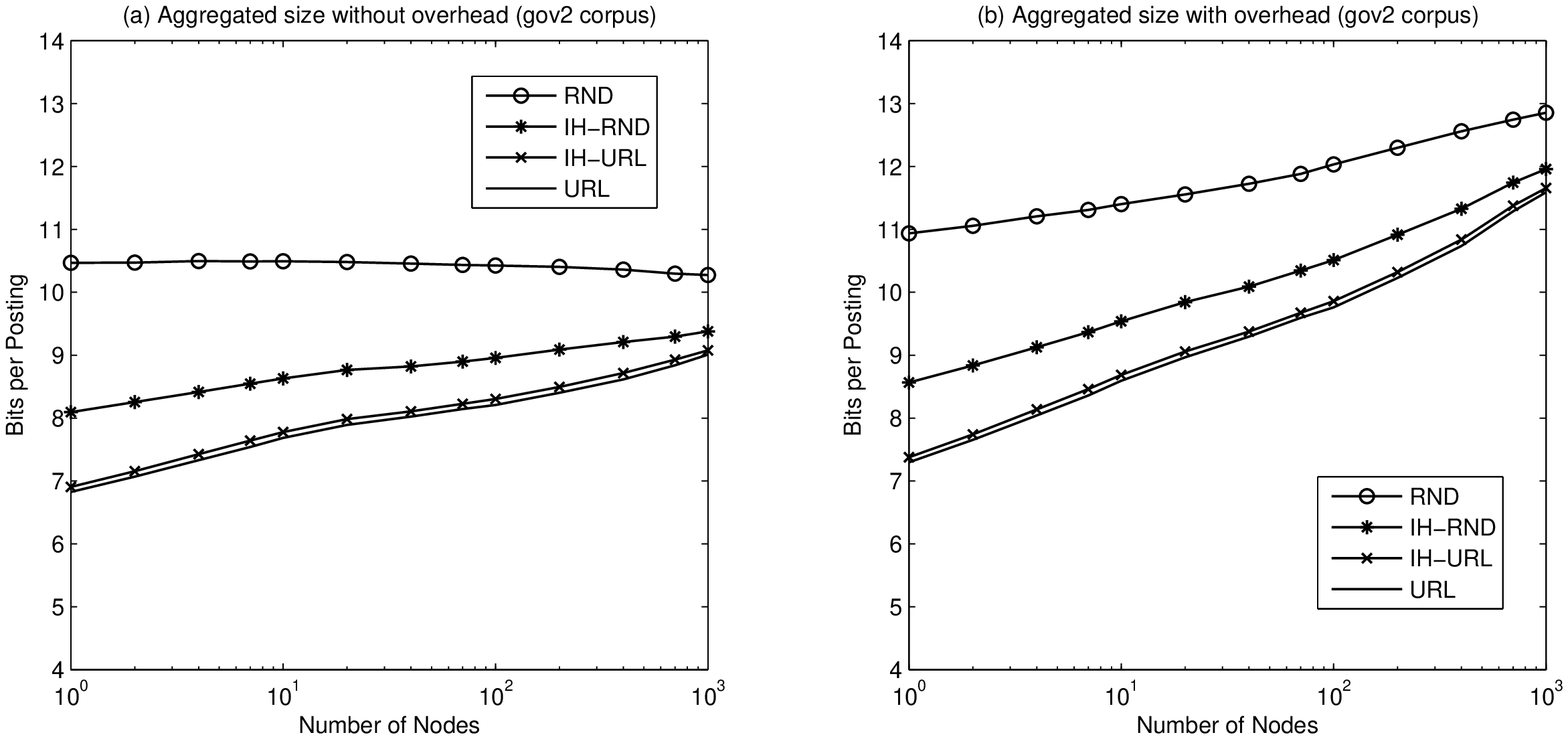}
\caption{Bits per posting as function of the number of nodes for different docId assignment schemes and PForDelta encoding applied to .gov2 corpus, (a) without and (b) with overhead.}
\label{fig: gov2 PForDelta}
\end{figure*}

The bits per posting measure is plotted as function of the number of nodes with and without overhead in Figures \ref{fig: gov2}.a and \ref{fig: gov2}.b, respectively. The curves are plotted for the URL, IH-URL, IH-RND, and RND docId assignment schemes using the full .gov2 corpus and Delta encoding. Figure \ref{fig: gov2}.a demonstrates that without overhead, the aggregated size decreases with the number of nodes for the RND assignment and increases for the URL bases schemes (i.e., URL, IH-URL, IH-RND). In particular, the ratio between the sizes of the RND and the URL assignments decreases from $2.2$ when no partitioning is applied, to $1.45$ when the corpus is partitioning over $m=10^3$ nodes. When the overhead is included, the sizes achieved by all schemes increase with the number of nodes, although the performance of URL based schemes degrades at a faster rate than that of the RND scheme. As can be seen, in the region of interest, the curves are approximately linear in $\log m$. Beyond this region, as the number of nodes increases, the URL based curves will coincide with that of the RND, and in the limit where each document is placed on a different node the number of bits per posting of all schemes go to one.

In Figures \ref{fig: gov2-3M}.a and \ref{fig: gov2-3M}.b, the bits per posting measure with and without overhead is plotted as function of the number of nodes, respectively. The curves are plotted for the URL, IH-URL, IH-RND, RND, and KSCAN-TSP docId assignment schemes using 3 million pages taken as a URL-continuous bulk from .gov2 corpus\footnote{A smaller corpus is used due to run time considerations of the KSCAN-TSP scheme.} and Delta encoding. It can be seen that the compression achieved by the KSCAN-TSP scheme behaves similarly to that of the URL based schemes, and increases with the number of nodes. We note that although KSCAN-TSP is expected to perform as the RND scheme in the limit, where the number of nodes is large, one could expect that KSCAN-TSP will degrade more gracefully with the number of nodes than URL. This is since KSCAN-TSP (and other state-of-the-art schemes) have an additional degree of freedom over URL sorting, in their ability to reorder the local documents after partitioning. However, as seen here, both URL and KSCAN-TSP degrade at similar rates.

Comparing the figures \ref{fig: gov2} and \ref{fig: gov2-3M} produced for the full .gov2 corpus and for the 3 million document sub-corpus respectively, using Delta encoding, reveals that 
the shapes of the curves and the relations between them are similar. This strengthens our conjecture that the same behavior also hold for web scale collections.

Turning to PForDelta encoding, Figures \ref{fig: gov2 PForDelta}.a and \ref{fig: gov2 PForDelta}.b plot the bits per posting measure as function of the number of nodes with and without overhead, respectively. The curves are plotted for the URL, IH-URL, IH-RND, and RND docId assignment schemes using the full .gov2 corpus and PForDelta encoding. In general, the trends visible for Delta encoding and all docId assignment schemes (Figure \ref{fig: gov2-3M}) are also visible here for the PForDelta encoding curves. Nevertheless, comparing figures \ref{fig: gov2} and \ref{fig: gov2 PForDelta} it is observed that while the RND assignment curve decreases in a lower rate than that of the Delta encoding curve, the URL based sorting curves are increasing in a higher rate than those of the Delta encoding\footnote{Also visible is the superiority of Delta encoding over the specific variant of the PForDelta scheme used here, 
which is consistent with the results reported in \cite{Yan-Ding-Suel-WWW2009}.}. 
\comment{
While the trends visible for the URL based sorting and Delta encoding are clearly visible also here for PForDelta encoding, a slightly different behavior is observed for the RND assignment. Without overhead, the aggregated size still decreases with the number of nodes for the RND assignment but in a lower rate compared to that of the Delta encoding curve.}

Another observation, visible in all figures, relates to the true nature of URL sorting. 
By merely clustering each host's documents together, the IH-RND scheme achieves 75\% to 85\% (for Delta encoding over the range of node numbers) of the performance improvement of URL sorting over random assignment. Moreover, the performance of IH-URL is almost identical to that of URL. To be precise, URL is slightly better than IH-URL (about $1\%$ on the average) over the range of node numbers. We conclude that the impressive effectiveness of URL sorting for Web corpora such as .gov2, stems mostly from the act of clustering documents of the same host together (i.e., IH-RND scheme). Keeping the lexicographical order within each host is the secondary contributor to the effectiveness of the URL scheme, and when combined with host clustering (i.e., IH-URL), it provides almost identical performance to that of URL sorting. On the other hand, keeping the lexicographical order across hosts has a negligible effect, and hosts can be placed randomly without degrading the URL scheme's effectiveness. These conclusions, while of little practical implication, provide some insight into the true nature of URL sorting.

We note that these results were all generated under random document distribution to nodes (see Section \ref{sec: Index-Partitioning}). Experimental results (not presented here) with round-robin distribution did not produce qualitatively different results. In addition, experimental results (also not presented here) reveal a small variance between multiple runs. Hence, the corpora used are large enough that self averaging is dominant. Hence, multi runs are redundant and all the presented results are of a single run experiments.

\section{Analytical Insight on Results}\label{sec: Analytical Insight on Results}

This section provides analytical and illustrative explanations to some of our experimental results. In particular, we prove that for random docId assignment and individual dGap logarithmic encoding (e.g., Delta encoding), the average aggregated index size (ignoring overhead) is a non-increasing function of the number of partitions. Conversely, for URL (and IH-URL) assignment and individual dGap logarithmic encoding, we demonstrate that partitioning increases the aggregated size. We note that the impact of index-partitioning on docId assignment and PForDelta encoding is much harder to explain since this encoding scheme works in blocks of dGaps, and is left for further study.

\subsubsection*{Index-Partitioning and Random docId Assignment}

Our model for index-partitioning under random docId assignment is as follows. Let there be $\abs{D}$ documents and $m$ nodes, and assume for simplicity that $m$ divides $\abs{D}$ and that the documents are evenly distributed across the nodes. We first draw uniformly at random (u.a.r.) a permutation $\pi$ over the documents, and then draw an equal partitioning (denoted by $g^m$) of $m$ sets of $\frac{\abs{D}}{m}$ documents each, also u.a.r. There are $\abs{D}! \paren{m!}^{-{\frac{\abs{D}}{m}}}$ such partitions. The document sets get assigned to the servers, with the internal order on each server respecting (being consisting with) $\pi$. We aim to prove that
the following expectation, denoted $\Delta^m$, is non-negative:
\[\Delta^m = E_{\pi, g} [ {\cal P}(\pi) - {\cal P}^m(\pi,g) ] \geq 0\ ,\]
where ${\cal P}(\pi)$ denotes the length of all postings lists when the documents are ordered by $\pi$ and indexed on a single node, and ${\cal P}^m(\pi,g)$ denotes the aggregated length of all postings lists when the documents are partitioned by $g$ into $m$ nodes, with the internal order in each node respecting $\pi$. Now,
\begin{equation*}
\begin{aligned}
\Delta^m &=\sum_g \frac{\paren{m!}^{\frac{\abs{D}}{m}}}{\abs{D}!} \sum_\pi \frac{1}{\abs{D}!}[ {\cal P}(\pi) - {\cal P}^m(\pi,g) ]\\
&= \frac{\paren{m!}^{\frac{\abs{D}}{m}}}{\paren{\abs{D}!}^2} \sum_g \sum_\pi [{\cal P}(\pi) - {\cal P}^m(\pi,g)]\ .
\end{aligned}
\end{equation*}

Looking at the inner sum, observe that for a fixed partition $g$ and every permutation $\pi$ there exists a single permutation $\tilde{\pi}$ that represents the concatenation of the $m$ partial permutations. Furthermore, for a fixed $g$, the mapping between $\pi$ and $\tilde{\pi}$ is 1:1 and onto. We now define the {\em m-slice} partitioning of a $\abs{D}$-sized permutation, denoted $g_m$, as the process of assigning the first $\frac{\abs{D}}{m}$ documents to the first node, and so on, until assigning the last $\frac{\abs{D}}{m}$ documents to the $m$'th node. By definition applying $g$ on $\pi$ is equivalent to applying $g_m$ on $\tilde{\pi}$, and so:
\begin{equation*}
\begin{aligned}
\Delta^m &= E_{\pi, g} [ {\cal P}(\pi) - {\cal P}^m(\pi,g) ]\\
&= \frac{\paren{m!}^{\frac{\abs{D}}{m}}}{\paren{\abs{D}!}^2} \sum_g \sum_\pi[ {\cal P}(\pi) - {\cal P}^m(\tilde{\pi},g_m) ]\ .
\end{aligned}
\end{equation*}
As $\pi$ goes over all $\abs{D}!$ permutations, so does $\tilde{\pi}$, and thus
\begin{equation*}
\begin{aligned}
\Delta^m &= \frac{\paren{m!}^{\frac{\abs{D}}{m}}}{\paren{\abs{D}!}^2} \sum_g \sum_\pi[ {\cal P}(\tilde{\pi}) - {\cal P}^m(\tilde{\pi},g_m) ]\\
&= \frac{1}{\abs{D}!} \sum_{\tilde{\pi}}[ {\cal P}(\tilde{\pi}) - {\cal P}^m(\tilde{\pi},g_m) ]\ .
\end{aligned}
\end{equation*}
To conclude the proof, we argue that $\forall \tilde{\pi},\ {\cal P}(\tilde{\pi}) - {\cal P}^m(\tilde{\pi},g_m) \geq 0$. Since the transformation involves only slicing, all intra-slice dGaps remain the same for the original and partitioned indexes (or slices), while dGaps bridging across slices are shorter within the partitioned indexes. In expectation, the bridging dGaps are halved by the slicing process, and assuming a logarithmic encoding function (e.g. Delta encoding), about 1 bit is gained on account of each bridging dGap. As the number of nodes (or slices) $m$ increases, more dGaps bridge across slices. Hence, the expected difference $\Delta^m$ does not decrease with $m$.  

\subsubsection*{Index-Partitioning and URL Sorting}

Ideally, a postings list following URL-based assignment includes runs of small dGaps separated by long dGaps. To illustrate the impact of index-partitioning into $m$ nodes on the performance of URL sorting, consider a specific posting list which begins with a single large dGap of $N_1$, followed by a run of $R$ dGaps of 1, another large dGap of $N_2$, and another run of $R$ dGaps of 1, with $N_1,\ N_2\gg R\gg m$. Under Delta encoding, the size of the postings list is $\delta(N_1)+\delta(N_2)+2R\delta(1)$. It is easily verified that the average aggregated size after partitioning is approximated by  $m[\delta(N_1/m)+\delta(N_2/m)+2(R/m)\delta(1)]$. Hence, the difference in the postings list sizes after and before partitioning into $m$ nodes is approximately $m [\delta(N_1/m)+\delta(N_2/m)]-(\delta(N_1)+\delta(N_2))$. 
Since Delta encoding behaves logarithmically, partitioning increases the average overall size by approximately $(m-1)(\log_2 N_1 + \log_2 N_2)-2m\log_2 m$. Obviously, this oversimplified example does not represent all cases, but it teaches us that for URL sorting (and IH-URL sorting), the encoding of the partitioned large dGaps of the original list causes its aggregated size to increase.



\section{Document Distribution Schemes and Query Processing Time}\label{sec:queryproc}

The previous sections demonstrated the deleterious effect of random distribution of documents to nodes, on the aggregated index sizes. This section examines the impact of document distribution schemes on other factors affecting query processing time, and demonstrates the significant benefits of random distribution -- which make it the industry standard \cite{Arasu-etal_ACM-Internet-Tech_2001,Barroso2003,IIR,prefetching}. In particular, we qualitatively show that random distribution results in faster query processing than that achieved by the better compressing URL-based distribution scheme.
\comment{
RL: what claim above is proved by slicing??
which provides better index sizes\footnote{This claim is immediately proved since slicing is shown in Section \ref{sec: Analytical Insight on Results} to reduce the aggregated index size.}.
}

\subsection{Surrogates for Query Processing Time}

In order for our ensuing experiments and qualitative analysis to be independent of specific retrieval algorithms or computational platforms, we use surrogate measures that are highly correlated with query evaluation time, for both \textit{disjunctive} and \textit{conjunctive} query models. In what follows, let $q=\{ t_1,\ldots,t_k\}$ be a $k$-term query, and let $\ell(t)$ denote the number of postings in term $t$'s postings list. In disjunctive queries, disregarding various pruning and early termination schemes, retrieval algorithms must scan all lists to fully evaluate the query. Hence, a surrogate measure for the running time of a disjunctive query on a particular index partition would be $\sum_{t \in q}\ell(t)$. In a locally-partitioned index among $m$ nodes, query evaluation (again, disregarding timeout or pruning policies) must wait for the slowest partition to finish evaluating the query. Hence, we approximate the running time of $q$ on $m$ nodes in disjunctive semantics, $\mathcal{T}_d(q)$, by 
\begin{equation*}
    \mathcal{T}_d(q) =\max_{j=1,\ldots,m}\sum_{t \in q} \ell_j(t)\ ,
\end{equation*}
where $\ell_j(t)$ denotes the length of $t$'s postings list on the $j$'th partition. Moving to conjunctive models, queries are typically evaluated by join-flavored algorithms \cite{WAND,Dups2006} that rely on the ability to skip portions of postings lists where matches are known not to exist \cite{witten1}\footnote{We ignore the small overhead that such skipping mechanisms add to the lengths of the postings lists.}. Therefore, our surrogate for $q$'s running time on a particular index partition is the length of the postings list of its rarest term, $\min_{t \in q}\ell(t)$. In a distributed setting, the slowest partition dictates that 
\begin{equation*}
    \mathcal{T}_c(q) =\max_{j=1,\ldots,m}\min_{t \in q} \ell_j(t)\ .
\end{equation*}
We stress that we do not claim that these measures {\em equal} query running times - only that for most retrieval algorithms on RAM-resident indexes, they represent reasonable surrogates that are correlated with running times. 

\subsection{Experimental Evaluation}

We again use the TREC .gov2 corpus (see Section \ref{sec: setup}), and distribute its documents to servers using random distribution (RND), and two flavors of URL-based distribution. First, vanilla URL distribution (URL), where all documents are ordered lexicographically according to their URL and then evenly sliced and routed to servers; second, IH-URL distribution - where hosts are randomly ordered and same host documents are lexicographically sorted according by URL before being evenly sliced and routed to servers. We use the 150 queries of TREC topics 701-850, whose average length is 3.1 terms, and report the average $\mathcal{T}_c=\frac{1}{150}\sum_{q=701}^{850} \mathcal{T}_c(q)$ and the similarly defined average $\mathcal{T}_d$ resulting from the three document distribution schemes over all queries, to qualitative compare their average query processing time.

Figure \ref{fig: query} plots the $\mathcal{T}_c$ and $\mathcal{T}_d$ curves for the two query types as functions of the number of servers, for the three document distribution schemes RND, URL, and IH-URL. The figure reveals the significant benefit of RND over the URL-based assignment schemes in terms of query processing time, and furthermore that the difference between the curves induced by RND and the URL-based schemes increases with the number of servers. In particular, RND induces $\mathcal{T}_c$ and $\mathcal{T}_d$ curves that are an order of a magnitude lower (i.e. faster) than those induced by the URL-based schemes at $m=1000$ servers. A closer inspection of the RND curves reveal that their slope is approximately $-1$ in a $log-log$ scale. Hence, RND scheme induced $\mathcal{T}_c$ and $\mathcal{T}_d$ are proportionally inverse to the number of servers: ${\mathcal{T}_c},{\mathcal{T}_d}\propto \frac{1}{m}$. Finally, note the similar performance demonstrated by the two URL-based schemes, which is explained by the weak inter-host document similarity already observed in Section \ref{sec: results}.

The degradation in query processing time obtained by the URL-based distribution schemes can be intuitively explained by the fact that same host documents are similar (which is good for reducing the index size) and share many terms. Hence, placing them on the same partition yields unbalanced posting lists which increases query processing time due to the maximum operation included in the calculation of both $\mathcal{T}_c$ and $\mathcal{T}_d$.

\begin{figure}[t]
\centering
\includegraphics[scale=0.55]{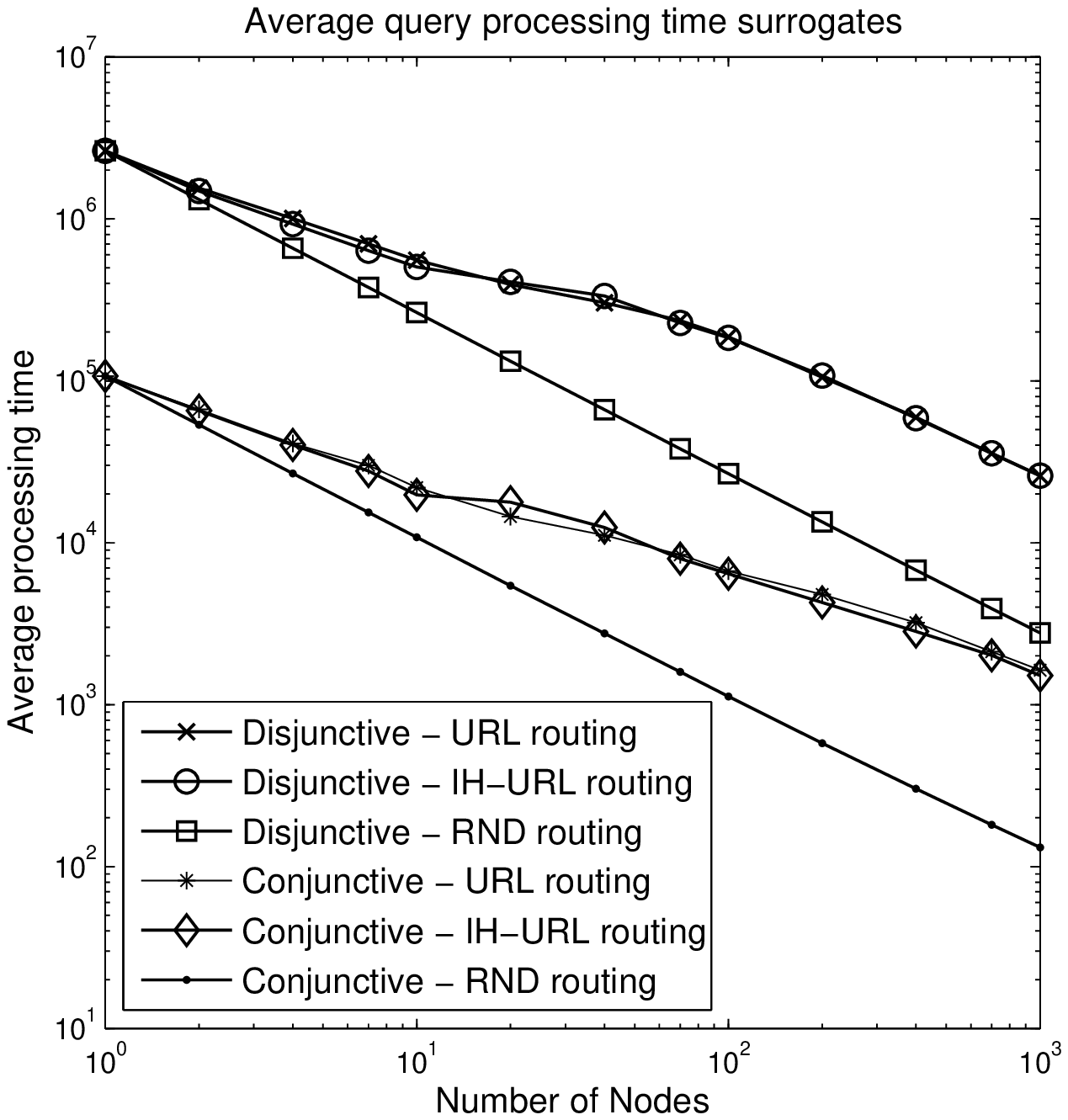}
\caption{Average query processing time surrogates vs. number of servers for different document routing schemes.}
\label{fig: query}
\end{figure}

\subsection{Analytical Evaluation}

This subsection analytically explains why the slopes of RND's $\mathcal{T}_c$ and $\mathcal{T}_d$ curves are inversely proportional to the number of servers $m$. For simplicity, we assume the document generation model of \cite{Chierichetti-Kumar-Raghavan-WWW2009}, in which terms are picked to document independently. Hence, for a disjunctive query $q$, we can equate $\mathcal{T}_d(q)$ to the most occupied among $m$ urns (servers) when $b_q=\sum_{t\in q} df(t)$ balls (postings) are randomly tossed to the urns ($df(t)$ denotes the document frequency of term $t$ in the entire corpus) \cite{Johnson-Kotz_Book1977}. 

\begin{prop}\label{prop: urn} For any $0 < \epsilon < 1$ and $\delta_\epsilon\le\delta < 2e-1$, 
\begin{equation*}
\mbox{Prob}\left(\ \mathcal{T}_d(q)\in\bracket{\frac{b_q}{m},
\frac{b_q}{m}\paren{1+\delta}}\ \right) > 1-\epsilon\ ,
\end{equation*}
\begin{equation*}
\mbox{where}\ \delta_\epsilon\triangleq \sqrt{\frac{4m}{b_q}\log\frac{m}{\epsilon}}\ .
\end{equation*}   
\end{prop}
\begin{proof}
Let $b_q$ balls be tossed randomly into $m$ urns, and let $x_j$ be the number of balls in the $j$th urn. In addition, denote the expected number of balls in an urn by $\mu\triangleq\frac{b_q}{m}$, and let $x_{max}=\max_{j=1\ldots,m}x_j$ be the maximal number of balls falling into some urn. Setting $\alpha\triangleq\mu\paren{1+\delta}$ for some $\delta>0$ we can write
\begin{equation*}
\begin{aligned}\label{eq: max urn 1}
Pr\paren{x_{max}\ge \alpha}&=Pr\paren{\bigcup_{j=1}^m x_j\ge \alpha}\\
&\le m\ Pr\paren{x_1\ge\alpha}\le m\ e^{-\frac{\mu\delta^2}{4}}\ ,
\end{aligned}
\end{equation*}
where the first inequality is due to the union bound, and the second inequality is achieved by applying Chernoff's bound and holds for $\delta < 2e-1$. Forcing the last term of the previous expression to be smaller than $0<\epsilon<1$, we have that $\delta$ must also satisfy 
\begin{equation*}
   \delta \geq \sqrt{\frac{4m}{b_q}\log{\frac{m}{\epsilon}}} \ .
\end{equation*}
The proof is completed by recalling that $x_{max}\ge \mu$.
\end{proof}

The average number of postings $b_q$ for the $N=150$ TREC topics 701-850 is about $2.6\times 10^6$, whereas the number of servers in this experiment does not exceed $10^3$. Therefore, we can apply Prop. \ref{prop: urn} and write
\begin{equation*}
    \mathcal{T}_d = \frac{1}{N}\sum_{q} \mathcal{T}_d(q) \approx \frac{1}{m\ N}\sum_{q} b_q = \frac{1}{m}\paren{\frac{1}{N}\sum_{q}\sum_{t\in q}df(t)}\ . 
\end{equation*}
Hence,  $\mathcal{T}_d$ is inversely proportional to the number of servers $m$, as observed in Fig. \ref{fig: query}.

\comment{
we have that for the dominant queries $\delta_0\le 0.135$ for $\epsilon=1\%$ and $m=1000$ ($\delta_0\le 0.04$ for $m=100$). Hence, we can apply Prop. \ref{prop: urn} and approximate the average query processing time by
\begin{equation*}
    \overline{\mathcal{T}}_d = \frac{1}{N}\sum_i \mathcal{T}_d(i) \approx \frac{1}{N\ m}\sum_i n(i) = \frac{1}{N\ m}\sum_i\sum_\ell t_\ell(i)\ . 
\end{equation*}
It is concluded that $\overline{\mathcal{T}}_d$ is proportionally inverse to the number of servers $m$, as observed in Fig. \ref{fig: query}.
}

\comment{Turning to the conjunctive queries we claim that the order of processing time for the $i$th query, $\mathcal{T}_c(i)$ equals to the maximum urn occupancy of the following urn model. According to this model there are colored balls (the number of colors equals the number of the query terms) where the number of balls of the $\ell$th color equals  the frequency of the $\ell$th term within the corpus $t_\ell(i)$. After the balls are randomly tossed into $m$ urns, we denote the color with least balls in each urn and remove all other balls. 
\begin{prop}\label{prop: con}
The maximum occupancy of the above urn model with $n_\ell$ balls of the $\ell$ 
th color and $m$ urns satisfies
\begin{equation*}
x_{max}\in\bracket{\frac{n}{m},
\frac{n}{m}\paren{1+\delta}}\ ,    
\end{equation*}
with probability larger than $1-\epsilon$ for any $0<\epsilon<1$ and
\begin{equation*}
    \delta \ge \delta_0\triangleq \sqrt{\frac{4m}{n}\log\frac{m}{\epsilon}}\ .
\end{equation*}   
\end{prop}
\begin{proof}
See Appendix \ref{app: urn}.
\end{proof}}

The expression $\mathcal{T}_c$ corresponding to conjunctive queries involves a max-min operation, which complicates the exact analysis. Therefore, we analyze an upper bound which is obtained by only considering the rarest term of each query.
In this case, $\mathcal{T}_c(q)$ equals the maximum urn occupancy of a simple urn model where $b_q=\min_{t \in q} df(t)$ balls are randomly tossed into $m$ urns. Since the average of $b_q$ for the $N=150$ queries of TREC topics 701-850 is about $10^5$ 
-- still at least two orders of magnitude over the number of servers $m$, we can apply Prop. \ref{prop: urn} and write
\begin{equation*}
    \mathcal{T}_c = \frac{1}{N}\sum_{q} \mathcal{T}_c(q) \approx \frac{1}{m\ N}\sum_{q} b_q = \frac{1}{m}\paren{\frac{1}{N}\sum_{q}\min_{t\in q}df(t)}\ . 
\end{equation*}
Hence,  as in the disjunctive case, $\mathcal{T}_c$ is inversely proportional to the number of servers $m$, as observed in Fig. \ref{fig: query}. It is noted that this approximated upper bound is tight since it has $-1$ slope in $log-log$ scale, and it includes the same constant as the experimented curve for $m=1$.

\comment{
we have that for the dominant queries $\delta_0\le 0.66$ for $\epsilon=1\%$ and $m=1000$ ($\delta_0\le 0.19$ for $m=100$). Hence, we can apply Prop. \ref{prop: urn} and approximate the average query processing time by
\begin{equation}
    \overline{\mathcal{T}}_c = \frac{1}{N}\sum_i \mathcal{T}_c(i) \approx \frac{1}{N\ m}\sum_i n(i) = \frac{1}{N\ m}\sum_i\min_\ell t_\ell(i)\ . 
\end{equation}
As with the disjunctive query case, it is concluded that $\overline{\mathcal{T}}_c$ is proportionally inverse to the number of servers $m$, as observed in Fig. \ref{fig: query}. It is noted that this approximated upper bound is tight since it has $-1$ slope in loglog scale and it includes the same constant as the experimented curve for $m=1$.

The degradation in query processing time obtained by the URL based partitioning schemes can be intuitively explained by the fact that same host documents are similar (which is good for reducing the index size) and share many terms. Hence, placing them on the same partition yields unbalanced posting lists which increases query processing time due to the maximum operation included in the calculation of both $\mathcal{T}_c$ and $\mathcal{T}_d$.    
}

\vspace{-0.1cm}
\section{Conclusions}\label{sec: conc}

We studied the impact of random index-partitioning on the performance of various docId assignment techniques, and demonstrated the deleterious effect of random index-partitioning in terms of the aggregated size of the partitioned index. 
We conjecture that our findings, based on the TREC .gov2 corpus and backed by some analysis, also hold at web scale - that randomized index-partitioning generates local collections that state-of-the-art ordering schemes can compress with relatively minor improvement over random ordering. The main reason for that, is that random index-partitioning causes pages of the same web host to be scattered over many nodes, resulting in local collections that are ``sparse'' in terms of URL continuity and that include few documents having high similarity with each other. Therefore, it follows that from a pure index size perspective, global index-partitioning where terms (instead of documents) are partitioned between nodes will compress better than the industry standard of randomized local index-partitioning. We also show via experimental evaluation that most of the effectiveness of URL sorting is achieved by merely clustering same host documents together. Moreover, we demonstrate that while URL sorting the documents within the hosts does yield additional improvement, keeping the lexical URL ordering of the hosts brings only negligible benefit. Lastly, we demonstrate the benefits of the industry standard random partitioning of documents to servers in terms of query processing time, over URL-based partitioning schemes.

\vspace{-0.1cm}





\end{document}